\newif\ifpublic
\publictrue 

\ifpublic
\documentclass[10pt,a4paper]{article}
\else
\documentclass[10pt,draft,a4paper]{article}
\fi

\usepackage[pagebackref,colorlinks=true,linkcolor=blue,urlcolor=blue,citecolor=blue,pdfstartview=FitH]{hyperref}
\usepackage[utf8]{inputenc} 
\usepackage[T1]{fontenc}    
\usepackage{amsmath}       
\usepackage{amssymb}
\usepackage{amsthm}
\usepackage{nicefrac}       
\usepackage{fullpage}
\usepackage{mathpazo}
\usepackage{mathtools}
\usepackage{thmtools}
\usepackage{thm-restate}

\ifpublic
\usepackage[disable]{todonotes}
\else
\usepackage[colorinlistoftodos]{todonotes}
\fi

\newcommand*\samethanks[1][\value{footnote}]{\footnotemark[#1]}

\usepackage[capitalize,nameinlink]{cleveref}
\renewcommand{\eqref}[1]{\hyperref[#1]{(\ref*{#1})}}

\theoremstyle{plain}
\newtheorem{theorem}{Theorem}[section]
\newtheorem{corollary}[theorem]{Corollary}
\newtheorem{lemma}[theorem]{Lemma}

\newtheorem{question}{Question}
\newtheorem{proposition}[theorem]{Proposition}
\newtheorem{definition}[theorem]{Definition}
\newtheorem{claim}[theorem]{Claim}

\theoremstyle{definition}
\newtheorem{remark}[theorem]{Remark}

\newcommand{\F}{\mathbb{F}}

\newcommand{\naturals}{\mathbb{N}}
\newcommand{\integers}{\mathbb{Z}_{\geq 0}}
\newcommand{\pintegers}{\mathbb{Z}_{> 0}}

\renewcommand{\epsilon}{\varepsilon}

\renewcommand{\phi}{\varphi}
\newcommand{\polylog}{\mathrm{polylog}}
\newcommand{\set}[1]{\{#1\}}

\DeclarePairedDelimiter\floor{\lfloor}{\rfloor}
\newcommand{\tc}{\mathsf{TC}}
\newcommand{\calC}{{\mathcal{C}}}
\newcommand{\tcz}{\mathsf{TC}_{\mathbb{Z}}}
\newcommand{\Lag}{\operatorname{Lag}}
\newcommand{\splt}{\operatorname{split}}
\newcommand{\wt}{\operatorname{wt}}

\title{A note on the explicit constructions of tree codes over polylogarithmic-sized alphabet}

\author{
  Siddharth Bhandari\thanks{Tata Institute of Fundamental Research,
    INDIA. email: {\tt
      \{siddharth.bhandari,prahladh\}@tifr.res.in}. Research of
    the authors supported by the Department of Atomic Energy,
    Government of India, under project no. 12-R\&D-TFR-5.01-0500. Research
    of the second author supported in part by the Swarnajayanti fellowship.}
   \and
 Prahladh Harsha\samethanks
}

\begin{document}
\maketitle

\begin{abstract}
Recently, Cohen, Haeupler and Schulman gave an explicit construction of binary tree codes over polylogarithmic-sized output alphabet based on Pudl\'{a}k's construction of maximum-distance-separable (MDS)  tree codes using totally-non-singular triangular matrices. In this short note, we give a unified and simpler presentation of Pudl\'{a}k and Cohen-Haeupler-Schulman's constructions. 

\end{abstract}

\section{Introduction}
We begin by recalling the definition of tree codes as introduced by Schulman in 1993~\cite{Schulman1996} 

\begin{definition}[tree codes]
\label{defn:TreeCode}
Given two finite alphabets $\Sigma$ and $\Gamma$ and a parameter $\delta \in (0,1)$, a tree code $\tc\colon\Sigma^{*} \to \Gamma^{*}$ with distance $\delta$ is a function mapping $n$-length strings  over the alphabet $\Sigma$ to $n$-length strings over the alphabet $\Gamma$ for every positive integer $n$ satisfying the following two properties.
\begin{description}
\item[Online encoding:] For every positive integer $m$, there exist functions $\phi_m\colon\Sigma^m \to \Gamma$ such that for all $n \in \naturals$, $x \in \Sigma^n$ and $i \in [n]$, we have $\tc(x)_i = \phi_i(x_1,\dots,x_i)$. In other words, the $i^{\text{th}}$ symbol of the output only depends on the first $i$ symbols of the input and not on the latter symbols. The functions $\phi_i$ are referred to as the encoding functions. 
\item[Distance:] The (relative) distance of the tree code $\tc$, denoted by $\delta_\tc$, (defined as follows) is at least $\delta$.
\[ \delta_\tc := \inf_{n, x\neq x' \in \Sigma^n} \frac{\Delta(\tc(x),\tc(x'))}{n-\splt(x,x')}, \]
where $\splt(x,x')$ is the largest integer $s$ such that for all $i \leq s$, we have $x_i = x'_i$. In other words, for every positive integer $n$ and any two distinct $n$-length strings $x, x' \in \Sigma^n$, we have that the Hamming distance between $y= \tc(x)$ and $y'=\tc(x')$ is at least $\delta (n-\splt(x,x'))$. 
\end{description}
The string $y=\tc(x)$ is said to be the tree code encoding of $x$. This encoding is said to be {\em explicit} if each of the encoding functions $\phi_i$ can be computed in time polynomial in $i$.  
\end{definition}

It is known via the probabilistic method that for every finite alphabet $\Sigma$ and distance parameter $\delta\in (0,1)$, there exists an alphabet $\Gamma = \Gamma(\Sigma,\delta)$ and a tree code $\tc\colon\Sigma^* \to \Gamma^*$ with distance at least $\delta$ (see \cref{thm:prob_tc} for the precise statement when $\Sigma$ is a finite field). A long standing open problem is to construct explicit tree codes matching the above construction (or even one for some distance parameter $\delta \in (0,1)$). Since we do not yet know such explicit constructions, several intermediate variants and special cases of the above definition of tree codes have been studied, some of which are interesting of their own right. Below, we mention some of these variants.

\begin{description}

\item[Truncated Tree Code:] Instead of constructing a single tree code $\tc\colon\Sigma^* \to \Gamma^*$ that works for all input lengths, sometimes it is easier to construct a family of tree codes (indexed by $n$) each of which works for all strings up to a fixed length $n$, i.e., $\tc^{(n)}\colon\Sigma^{(\leq n)} \to \Gamma^{(\leq n)}$.  We will refer to these tree-codes as {\em $n$-truncated} tree codes and the family $\{\tc^{(n)}\}_{n\in \pintegers}$ as truncated tree codes. While dealing with $n$-truncated tree codes, we will use the superscript $n$ to denote the maximum input length. The notion of distance now becomes:
\[
\delta_{\tc^{(n)}}:= \inf_{i\leq n, x\neq x' \in \Sigma^i} \frac{\Delta(\tc(x),\tc(x'))}{i-\splt(x,x')}
\]
Note that for a truncated code $\{\tc^{(n)}\}_n$, we have a family of encoding functions $\{\phi^{(n)}_i : 1 \leq i \leq n\}$ for each $n$. In this case, we say that the truncated code is {\em explicit} if for each $n$ and $i \leq n$, $\phi^{(n)}_i$ is computable in time polynomial in $n$. Note that it is polynomial in $n$ and not just $i$ as in the case of (non-truncated) tree codes. Clearly, tree codes imply truncated tree codes but the other direction is not known.

\item[Input Alphabet:] The classical tree-code considers binary input alphabet, i.e., $\Sigma = \{0,1\}$. Sometimes, it is convenient to have an algebraic structure on the input alphabet, say $\Sigma$ is a finite field $\F$ or even the set of non-negative integers $\integers$. When dealing with truncated tree-codes $\tc^{(n)}$, one may also consider tree codes whose input and output alphabets grow with $n$ (i.e, $\{\Sigma_n\}_n$, $\{\Gamma_n\}_n$ corresponding to $\{\tc^{(n)}\}_n$). 

\item[Output Alphabet:] The holy grail is to construct tree codes over the binary alphabet with constant output alphabet size. As an intermediate step, one typically constructs tree codes whose output alphabets grow with $n$. This is best formalized in the setting of truncated tree codes $\{\tc^{(n)}\colon \Sigma^{(\leq n)} \to \Gamma_n^{(\leq n)} \}_{n \in \integers}$ where $\{\Gamma_n\}_n$ is a family of output alphabets whose size (may) grow with $n$.  The formalism for the more standard (single) tree code is a little more involved. We let $\Gamma = \cup_{n} \Gamma_n$ and $\tc\colon\Sigma^* \to \Gamma^*$ subject to the constraint that the encoding functions $\phi_i$ satisfy $\phi_i\colon \Sigma^i \to \Gamma_i$. Note $\Gamma$ is not necessarily a finite alphabet, though each $\Gamma_n$ is. 
\end{description}

We are now ready to state the remarkable recent result of Cohen, Haeupler and Schulman that constructs explicit binary tree codes over a polylogarithmic-sized output alphabet.
\footnote{The theorem as stated here is marginally stronger than the one proved by Cohen-Haeupler-Schulman, in the sense that it refers to explicit constructions of tree codes as opposed to truncated codes. See \cref{rem:honesttogodtreecode} for more details.}

\begin{restatable}[explicit tree codes with polylog alphabet~{\cite{CohenHS2018}}]{theorem}{binarytreecodearbitrarydistance}
\label{thm:binarytreecodearbitrarydistance}
For all $\eta \in [0,1)$ there exists a family of alphabets $\{\Gamma_n\}_{n}$ such that
$|\Gamma_n|= O_\eta(\polylog(n))$ with the following property. Let $\Gamma=
\cup_n \Gamma_n$. There exists a tree code $\tc\colon \{0,1\}^* \to
\Gamma^*$ with distance $\eta$ such that the encoding functions $\phi_i$ satisfy
satisfy $\phi_i\colon \{0,1\}^i \to \Gamma_i$. Furthermore, the encoding functions $\phi_n$ are
constructible in time polynomial in $n$.
\end{restatable}


Cohen, Haeupler and Schulam construct such tree codes in two steps. In the first step, they construct  tree codes over the integer alphabet using Newton polynomials. This construction is a special case of Pudl\'ak's construction of "maximum-distance-separable (MDS)" tree codes using totally-non-singular triangular matrices. In the second step, they perform an alphabet reduction to reduce the alphabet from integers to a polylogarithmic-sized output alphabet. 

\paragraph{Organization:} The rest of this exposition is organized as follows. We give a short (but complete) exposition of MDS codes, a la Pudl\'{a}k, in \cref{sec:mds} followed by the alphabet reduction technique of Cohen-Haeupler-Schulman in \cref{sec:ared}.

The purpose of this note is only to give an exposition of the Cohen-Haeupler-Schulman construction. In particular, we place the Cohen-Haeupler-Schulman construction in the context of Pudl\'{a}k's framework for constructing explicit tree codes. For a more comprehensive treatment of tree codes, we refer the interested reader to Gelles's survey on interactive communication~\cite{Gelles2017} and Pudl\'ak's excellent paper~\cite{Pudlak2016} discussing various potential approaches to obtain explicit tree codes.

\subsection{Linear tree codes}\label{sec:linear}

A particularly appealing setting is when the input and output alphabet $\Sigma,\Gamma$ are finite-dimensional vector spaces $\F^s, \F^r$ respectively over some finite field $\F$. In this setting, if the tree code $\tc\colon (\F^s)^* \to (\F^r)^*$ further satisfies that $\tc|_{(\F^s)^n}$ is a linear mapping for each $n$, then the tree code is said to {\em $\F$-linear}. 

For the case when $s =1$, it is easy to check that the online encoding property of tree codes ensures that if $\tc\colon\F^* \to (\F^r)^*$ is $\F$-linear, then there exist $r$-lower triangular matrices $A_0, \dots, A_{r-1} \in \F^{\pintegers \times \pintegers}$ such that for any $x \in \F^n$, we have $\tc(x)_i = ((A^{(n)}_0(x))_i, (A^{(n)}_1(x))_i, \dots, (A^{(n)}_{r-1}(x))_i) \in \F^r$ where $A^{(n)}$ refers to the $n\times n$-matrix formed by the restricting $A$ to the top $n$ columns and rows.  By a change of basis of the input strings, we may without loss of generality assume that $A_0 = I$, the identity matrix (we will not use this property, however our construction of distance-half linear tree codes will have this property). 

As in linear codes, the distance of a linear tree code can be characterized by the minimum weight of non-zero codewords. More specifically, if $\tc\colon\F^* \to (\F^r)^*$ is $\F$-linear then 
\[ \delta_{\tc} = \inf_{n, x \in \F^n \setminus \{{0^n}\}} \frac{\wt_{\F^r}(\tc(x))}{n-\splt(x,{0^n})},\]
where $\wt_\Gamma(x)$ refers to the Hamming weight of $x$ with respect to the alphabet $\Gamma$. Since $\Gamma= \F^r$ and $\tc(x) \in (\F^r)^n = \F^{rn}$, it is also natural to consider a variation of the above definition with Hamming weight with respect to the finite field alphabet $\F$ instead.
\[ \widetilde{\delta}_{\tc} := \inf_{n, x \in \F^n \setminus \{{0^n}\}} \frac{\wt_{\F}(\tc(x))}{r(n-\splt(x,{0^n}))}.\]
Clearly, $\widetilde{\delta}_\tc \leq \delta_\tc$. 

Schulman proved the existence of tree codes via the probabilistic method. The following statement is due to Pudl\'ak~\cite[Theorem2.1]{Pudlak2016}~\footnote{Pudl\'ak provides two probabilistic constructions, one for truncated codes and another for untruncated codes. The theorem stated here refers to the latter construction using cyclic Toeplitz matrices.}.
\begin{theorem}[probabilistic construction of linear tree codes{~\cite{Schulman1996, Pudlak2016}}]
\label{thm:prob_tc}
  Let $q=|\F|$, $r=q^d=|\Sigma|$ and $0<\delta<\frac{r-1}r$
  such that  
\begin{align*}
    \log_r(2q)+H_r(\delta)\leq 1 
\end{align*}
where $H_r$ denotes the $r$-entropy function defined by $$H_r(x)=x\log_r(r-1)-x\log_rx-(1-x)\log_r(1-x).$$

  Then, there exists a linear tree code $\tc:\F^{*}\to\Sigma^{*}$ with $\widetilde{\delta}_\tc >
  \delta$. Moreover, if $q,r$ and $\delta$ are fixed, then the encoding functions $\phi_n$
  can be constructed for every $n$ in time $2^{O(n)}$.
\end{theorem}

Note that there exists $\delta>0$ such that for every $q>2$, there
exist tree codes with rate $1/2$ (i.e., $d=2$) and minimum relative
distance $\geq\delta$. We do not know if binary (i.e., $q=2$) tree codes
with rate $1/2$ can have asymptotically positive minimum
relative distance.  

\section{MDS tree codes, a la Pudl\'{a}k}
\label{sec:mds}

In this section, we present the treatment of linear tree codes following Pudl\'ak~\cite{Pudlak2016}. More precisely, we present the ``Singleton bound'' for linear codes, define maximum-distance-separable (MDS) tree codes and present a construction of MDS codes assuming the existence of totally-non-singular triangular matrices.

\begin{proposition}[Singleton bound for tree codes~{\cite[Proposition~5.1]{Pudlak2016}}]\label{prop:singletonbound}
For every $n$-truncated code, $\tc^{(n)} \colon \Sigma^{(\leq n)} \to \Gamma^{(\leq n)}$ we have $\delta_{\tc^{(n)}} \leq \frac1n\left\lfloor n\left(1-\nicefrac{\log|\Sigma|}{\log|\Gamma|}\right)+1\right\rfloor$. 
\end{proposition}

\begin{proof}
This proof will be an adaptation of the standard proof of the Singleton bound for codes to the tree code setting. The proposition is trivially true when $|\Gamma| \geq |\Sigma^n|$ (in fact, in this case, there exist $n$-truncated tree codes with distance 1). So we might as well assume that $|\Gamma| < |\Sigma^n|$. Let $\ell$ be the smallest integer such that $\ell \geq  n\cdot(\nicefrac{\log |\Sigma|}{\log|\Gamma|})- 1 $, which is at least 1 since $n\cdot(\nicefrac{\log |\Sigma|}{\log|\Gamma|}) >1$. Consider the mapping $f_{n\to \ell}\colon \Sigma^n \to \Gamma^\ell$ defined as follows $f_{n \to \ell}(x) := \tc^{(n)}(x)|_{[n-\ell+1,n]}$ (i.e, the tree code encoding restricted to the last $\ell$ symbols). Since $|\Sigma^n| > |\Gamma^\ell|$, there must exist two strings $x,x' \in \Sigma^n$ such that $f_{n \to \ell}(x)= f_{n\to\ell}(x')$ (i.e., the tree code encodings of $x$ and $x'$ agree on the last $\ell$ symbols).
 Let $i\geq 1$ be the first place where $x$ and $x'$ differ. Note that $i \leq n - \ell$. Due to the online nature of tree codes we know that $\tc^{(n)}(x)$ and $\tc^{(n)}(x')$ agree on the first $i-1$ coordinates (in addition to the last $\ell$ coordinates).  
 Therefore, 
 \[
 \delta_{\tc^{(n)}} \leq \frac{n-(\ell+(i-1))}{n-i+1} =1-\frac{\ell}{n-i+1}\leq 1-\frac{\ell}{n} = 1-\frac1n\left(\left\lceil n \cdot \frac{\log|\Sigma|}{\log|\Gamma|} \right\rceil -1\right) = \frac1n\left\lfloor n\left(1-\frac{\log|\Sigma|}{\log|\Gamma|}\right)+1\right\rfloor.
\]
\end{proof}
Notice that the aforementioned proof would have also worked with the following relaxed notion of distance (that only considers strings of length $n$):
\[
\delta_{\tc^{(n)}} ^\text{ relaxed}:= \inf_{x\neq x' \in \Sigma^n} \frac{\Delta(\tc(x),\tc(x'))}{n-\splt(x,x')}.
\]

For any non-zero $y \in \Gamma^n$ and $\alpha \in (0,1)$, the condition $\wt(y) \geq \left\lfloor \alpha \cdot n+1\right\rfloor$ is equivalent to the condition that $\wt(y) > \alpha \cdot n$. This combined with  \cref{prop:singletonbound} motivates the following definition. 

\begin{definition}[MDS tree codes]\label{defn:MDSTreeCodes}
A $n$-truncated tree code $\tc^{(n)} \colon \Sigma^{(\leq n)} \to \Gamma^{(\leq n)}$ is said to be maximum-distance separable (MDS) if $\delta_{\tc^{(n)}} > 1- \nicefrac{\log |\Sigma|}{\log |\Gamma|}$. An MDS tree code is defined similarly. 
\end{definition}

Pudl\'ak showed how the existence of certain type of matrices implies the existence of MDS codes. We first define these special type of matrices. In the following, for a $n\times n$ matrix $A$ and subsets $I, J \subseteq [n]$ of the rows and columns, the matrix $A[I|J]$ refers to the submatrix of $A$ restricted to the rows $I$ and columns $J$. 

\begin{definition}[totally-non-singular triangular matrix]\label{defn:TraingularTotallyNonSingular}
A lower-triangular matrix $A \in \F^{n \times n}$ is said to be totally-non-singular if for every $1 \leq r \leq n$ and pair of sets $I, J \subseteq [n]$ such that $|I| = |J| =r $ and $I = \{1\leq i_1 < \cdots i_r \leq n\}$, $J = \{1\leq j_1 < \cdots j_r \leq n\}$ satisfying $i_s \geq j_s$ for all $s \in [r]$, we have that the submatrix $A[I|J]$ is non-singular.
\end{definition}

Given a lower-triangular matrix $A\in \F^{n\times n}$, consider the linear tree code $\tc^{(n)}_A\colon \F^{(\leq n)} \to (\F^2)^{\leq n}$ defined by the pair of generator matrices $(A_0 = I_n, A_1 = A)$. In other words, for any $x \in \F^k$ and $i \leq k$, we have 
\[\left(\tc^{(n)}_A(x)\right)_i := \left(x_i, \left(A^{(k)}(x)\right)_i\right).\]

\begin{theorem}[MDS linear tree codes~{\cite[Theorem~5.3]{Pudlak2016}}]
\label{thm:lowertiandgularcode}
If $A \in \F^{n\times n}$ is a totally-non-singular lower-triangular matrix, then the $n$-truncated tree code $\tc^{(n)}_A$ has distance greater than $\nicefrac12$ and is hence MDS.
\end{theorem}

\begin{proof}
We will in fact show the stronger property that $\widetilde{\delta}  > \nicefrac{1}{2}$. More precisely, for any non-zero codeword $x\in \F^k$ with $\ell=\splt(0^k,x)\in [0,k-1]$, we need to show that the $\F$-weight of $y = \tc(x) \in (\F^2)^k$ when viewed as a string in $\F^{2k}$ is greater than $(k-\ell)$. To this end, we define the following two sets of column and row indices. 
 \[\mathcal{C}_x :=\{j\in [k]\colon x_j\neq 0\}, \qquad \mathcal{R}_x :=\{i\in [k]\colon i> \ell, (A^{(k)}(x))_i=0\}.\] 

 The number of non-zero coordinates in $y$ (when viewed a string over $\F$) is exactly $|\mathcal{C}_x| + (k-\ell-|\mathcal{R}_x|)$ which is at strictly greater than $(k-\ell)$ if $|\mathcal{C}_x| > |\mathcal{R}_x|$, which is proved in the claim below.
 \end{proof}
 \begin{claim}\label{clm:cgeqr} $|\mathcal{C}_x| > |\mathcal{R}_x|$\end{claim}
 \begin{proof}
   Assume otherwise and suppose that $\mathcal{C}_x=\{j_1< j_2< \ldots< j_r\}$ and $\mathcal{R}_x=\{i_1< j_2< \ldots < i_s\}$ for some  $s\geq r$. By definition of $\mathcal{C}_x$, we have $j_1 = \ell+1$. Since $|\mathcal{R}_x| \geq |\mathcal{C}_x| \geq 1$, we have that $i_1 > \ell$ and hence $i_1 \geq j_1$.  Let $t$ be the maximum integer such that $i_1\geq j_1, i_2\geq j_2, \ldots, i_t\geq j_t$. Let us observe the action of the matrix $A[\{i_1, \ldots, i_t\}\mid \{j_1, \ldots, j_t\}]$ on the vector $x|_{\{j_1, \ldots, j_t\}}$. By maximality of $t$, we get that either $r=t$ or $i_{t+1} < j_{t+1}$. In either case, it is easy to see that $A[\{i_1, \ldots, i_t\}\mid \{j_1, \ldots, j_t\}](x|_{\{j_1, \ldots, j_t\}})=\mathbf{0}$. However, this is contradiction to the total non-singularity of $A$. Therefore, $|\mathcal{C}_x|>|\mathcal{R}_x|$.
   \end{proof}

\cref{clm:cgeqr} lends itself to the following alternate construction of a tree code which also has distance greater than $\nicefrac12$. For $v\in \F$ define $v'=(v,0) \in \F^{2}$ as the vector obtained by appending a zero to $v$. Similarly, for a vector $x\in \F^k$ let $x'\in (\F^2)^k$ be $(x_1',x_2',\ldots,x_k')$. Given a  totally-non-singular  lower-triangular matrix $A \in \F^{2n \times 2n}$, consider the $n$-truncated tree code $\tc^{(n)}_{A,(1,1)} \colon \F^{(\leq n)} \to (\F^2)^{(\leq n)}$ that maps $x \in \F^k$ to $A^{(2k)}(x')\in \F^{2k}$.  For a non-zero $x \in \F^k$, let $\ell =\splt(0^k,x)$. Let $x' \in \F^{2k}$ be the corresponding non-zero vector and let $\ell' = \splt(0^{2k},x)$. Clearly, $\ell'= 2\ell$. Let $y' = \tc^{(n)}_{A,(1,1)}(x')$ Define sets $\mathcal{C}_{x'}$ and $\mathcal{R}_{x'}$ as before. By definition, $|\mathcal{C}_{x'}| \leq \nicefrac12\cdot (2k-\ell') = k-\ell$ and $\wt_{\F}(y') = 2k -\ell' - |\mathcal{R}_x'| > 2k-\ell' - |\mathcal{C}_{x'}|$ (by \cref{clm:cgeqr}). Hence, $\wt_\F(y') > 2k - \ell' -(k-\ell) = k-\ell = \nicefrac12(2k-\ell')$. Hence, $\widetilde{\delta}$ distance of $\tc^{(n)}_{A,(1,1)}$ is greater than $\nicefrac12$.

This alternate construction has the advantage that it can be generalized to obtain MDS tree codes with distance $\delta_\tc> \nicefrac{r}{r+s}$ for integers $r,s\geq 1$.  Let $A\in \F^{(r+s)n\times (r+s)n}$ be a totally-non-singular  lower-triangular matrix. For $v\in \F^s$ define $v'\in \F^{r+s}$ as the vector obtained by appending $r$ zeros to $v$. Similarly, for a vector $x\in (\F^s)^k$ let $x'\in \F^{(r+s)k}$ be $(x_1',x_2',\ldots,x_k')$. 
For all $x\in (\F^s)^k$ and $k\leq n$ define  $\tc^{(n)}_{A,(s,r)}\colon(\F^s)^{(\leq n)}\to (\F^{r+s})^{(\leq n)}$ as
\[ x \in (\F^s)^k \longmapsto A^{(r+s)k}x' \in (\F^{(r+s)})^k. \]
An argument identical to the one above proves that this is an MDS code.

\begin{theorem}
\label{thm:MDScodes_r/(r+s)}
$\tc^{(n)}_{A,(s,r)}$ is a $n$-truncated MDS tree code with distance $\nicefrac{r}{r+s}$.
\end{theorem}
 

\subsection{Instantiating the MDS tree codes}

Pudl\'ak suggested various examples of triangular totally-nonsingular
matrices. One of these examples is the the Pascal Matrix defined as
follows:

\begin{definition}[Pascal matrix]\label{defn:PascalMatrix}
  The Pascal matrix $P \in
  \integers^{\integers \times \integers}$ is the matrix whose
  $(i,j)$-entry for $i,j \in \integers$ is the binomial coefficient $\binom{i}{j}$. For any $n
  \in \pintegers$, we let
  $P_n$ be the restriction of the (infinite) matrix $P$ to the first
  $(n+1)$ rows and $(n+1)$ columns. We will use the fact that for any $n$, the
  maximum value of any entry in $P_n$ is at most $2^n$.
\end{definition}

Gessel and Viennot~\cite{GesselV1985}, using Lindstr\"{o}m's combinatorial interpretation of certain determinants in terms of disjoint path systems~\cite{Lindstrom1973}, proved the following remarkable theorem which immediately yields the total-non-singularity of the Pascal matrix.

\begin{lemma}[~\cite{GesselV1985}, Corollary 2] 
Let $0\leq a_1< a_2< a_3<\ldots < a_n$ and $0\leq b_1< b_2< \ldots <b_n$ be two sequences of $n$ integers each. Define the $n\times n$ matrix $M$ by $M_{i,j}=\binom{a_i}{b_j}$. If $a_i\geq b_i$ for each $i\in [n]$ then $\det(M)>0$. 
\end{lemma}

Combining the above theorem with \cref{thm:lowertiandgularcode}, we
obtain the following result.

\begin{theorem}[tree code over integers~{\cite[Theorem~1.3]{CohenHS2018}}]\footnote{Cohen-Haeupler-Schulman obtained this theorem by using the inverse of the Pascal matrix instead. If a lower-triangular is totally-non-singular, then so is its inverse.}
\label{thm:inttreecode}
There exists an explicit tree code over the integers
$\tcz\colon \integers^{\integers} \to (\integers^2)^{\integers}$
with distance at least $1/2$ satisfying the property that for any $n
\in \pintegers$ and $i \in \{0,1\dots, n\}$, we have
$((\tcz(a_0,\dots, a_n))_i)_1 = a_i$ and  $((\tcz(a_0,\dots, a_n))_i)_2 \leq 2^n \max_{j} |a_j|$. 
\end{theorem}

\section{Alphabet Reduction}\label{sec:ared}

In this section, we show how to convert the tree code over the
integers constructed in \cref{thm:inttreecode} to one over the binary
alphabet. By merely restricting the input alphabet of $\tcz$ in \cref{thm:inttreecode} to integers within the range $\{0,1\dots, 2^{s-1}\}$ and setting $n = s$, we obtain the following
corollary.

\begin{corollary}[tree code over large alphabet]
\label{cor:largealph}
For every $s \in \pintegers$, there exist an explicit tree code over the
input alphabet $\{0,1\}^{s}$, $\tc^{(s)}\colon \{\{0,1\}^s\}^{(\leq s)}
\to \{\{0,1\}^{3s}\}^{(\leq s)}$ with distance at least $1/2$.
\end{corollary}

The tree code constructed in the above corollary has excellent
distance albeit over a large super-constant alphabet (both with
respect to input and output). In the rest of this section, we will
show how to transform this tree code into one over the binary alphabet
and polylogarithmic output alphabet. 

We follow, roughly, the treatment of Cohen-Haeupler-Schulman~\cite{CohenHS2018}. (Also, see Pudl\'{a}k~\cite[Proposition~3.1]{Pudlak2016} for the standard well-known alphabet reduction procedure from a polynomial-sized output alphabet to constant-sized output alphabet. In some sense, the Cohen-Haeupler-Schulman reduction is an extension of this procedure to super-polynomial-sized alphabets.)

To this end, we define the notion of a lagged tree code. 

\begin{definition}[tree codes with lag]
Given positive integers $\ell \leq L$, a tree code
$\tc\colon\Sigma^* \to \Gamma^{*}$ is said to have $\ell$-lagged
distance $\delta$ if the (relative) $\ell$-lagged distance
$\delta_\tc^{\Lag \ell}$ (defined below) is at least $\delta$. It is said
to have $(\ell,L)$-lagged
distance $\delta$ if the (relative) $(\ell,L)$-lagged distance
$\delta_\tc^{\Lag (\ell,L)}$ (defined below) is at least $\delta$.
\begin{align*}
\delta_\tc^{\Lag \ell} & := \inf_{n, x\neq x' \in
                         \Sigma^n,n-\splt(x,x') \geq \ell}
             \frac{\Delta(\tc(x),\tc(x'))}{n-\splt(x,x')},\\
\delta_\tc^{\Lag (\ell,L)} & := \inf_{n, x\neq x' \in \Sigma^n, \ell
                             \leq n-\splt(x,x') \leq L}
             \frac{\Delta(\tc(x),\tc(x'))}{n-\splt(x,x')}.
\end{align*}
The truncated versions of these lagged codes are defined appropriately.
\end{definition}

In other words, lagged distance is a relaxation of the standard notion
of distance when we consider only pairs of strings $x,x'$ whose
$\splt(x,x')$ lies in a particular interval.

The rest of the construction proceeds as follows. In step 1, we show that composing the integer tree code in \cref{cor:largealph} with a standard error-correcting code yields a truncated tree code with constant $\ell$-lagged distance for some $\ell = \ell(s)$. In step 2, we show that we can then obtain an untruncated code with constant $(\ell, L)$-lagged distance by placing interleaving copies of the truncated code obtained in step 1 along the infinite tree for some fixed function $L = L(\ell)$. Finally, in step 3, we show show that a superimposition of several $(\ell,L)$-lagged distance codes for increasing values of $\ell$ and $L$ yields our final code. The size of the final alphabet is determined by how fast the  function $L=L(\ell)$ grows which is in turn determined by the function $\ell = \ell(s)$.

We begin by stating some basic error-correcting code that we need for step 1.

\begin{proposition}\label{prop:ecc}
For all $\delta\in[0,1)$ there are constants $s_\delta$ and $c_\delta$ such that for all $s\geq s_\delta$ there is an explicit code $\mathcal{C}\colon\set{0,1}^{3s}\to (\set{0,1}^{c_\delta})^s$ with (relative) distance  at least $\delta$ . 
\end{proposition}

\begin{proof}
  Using standard error-correcting codes, we know that for every $\delta \in [0,1)$, there are constants $n_\delta\in \pintegers$, $\alpha_\delta \in (0,1)$ and alphabet $\Sigma_\delta$ such that there exist explicit codes $\Tilde{\mathcal{C}}\colon\set{0,1}^{n}\to \Sigma_\delta^{n/\alpha_\delta}$ for all $n \geq n_\delta$ with distance $\delta$ over the alphabet $\Sigma_\delta$\footnote{There are several ways to obtain such codes. Take a constant-rate-constant-distance code and distance amplify it using the ABNNR distance amplification technique. Or take Algebraic Geometric codes.}.

  Set $s_\delta := n_\delta/3$ and for $s \geq s_\delta$, let $\Tilde{\mathcal{C}}\colon\set{0,1}^{3s}\to \Sigma_\delta^{3s/\alpha_\delta}$ be the above code for $n=3s$. Choose $c_\delta$ such that $2^{c_\delta} \geq |\Sigma_\delta|^{3/\delta}$ so that there exists an injective mapping $\Sigma_\delta^{3/\alpha_\delta} \hookrightarrow \{0,1\}^{c_\delta}$. Under this mapping, the code $\Tilde{\mathcal{C}}\colon\set{0,1}^{3s}\to \Sigma_\delta^{3s/\alpha_\delta}$ can be viewed as $\mathcal{C}\colon\set{0,1}^{3s}\to (\set{0,1}^{c_\delta})^s$. Clearly, the distance of $\mathcal{C}$ is at least that of $\Tilde{\mathcal{C}}$.
\end{proof}

\paragraph{Step 1: truncated code with constant $\ell$-lagged distance:} We describe the truncated-lagged tree code \[\tc^{(s)}_{\Lag \ell}\colon \{0,1\}^{(\leq s^2)}
\to \{\{0,1\}^{c_\delta}\}^{(\leq s^2)}\] that results by composing the tree code, $\tc^{(s)}\colon(\set{0,1}^s)^{(\leq s)}\to (\set{0,1}^{3s})^{(\leq s)}$, in \cref{cor:largealph}, which is a distance half tree code, with the code $\calC$ from \cref{prop:ecc}. 

Formally, let $\calC$ be the code from \cref{prop:ecc} with output alphabet
$\set{0,1}^{c_\delta}$ and distance $\delta$ and $a\in \naturals$ a large constant to be determined later.
For every $s \geq s_\delta$ and $\ell\geq as$\footnote{We think of $a$ as a constant and of $\ell$ and $s$ as growing.}
we describe an explicit truncated-$\ell$-lagged tree code over the
binary input alphabet, $\tc^{(s)}_{\Lag \ell}\colon \{0,1\}^{(\leq s^2)}
\to \{\{0,1\}^{c_\delta}\}^{(\leq s^2)}$ with $\ell$-lagged distance at least $\delta(\nicefrac{1}{2}-\nicefrac{3}{2a})$.

The encoding process is a follows. Consider any $x\in \set{0,1}^k$ and let $x'= (x'_1,\dots,x'_{\lfloor k/s\rfloor})\in(\set{0,1}^s)^{\floor*{k/s}}$ obtained by grouping together blocks of $s$ symbols and ignoring any leftover symbols. Let $y'=(y'_1,\dots,y'_{\lfloor k/s\rfloor})=\tc^{(s)}(x')$, so $y'\in (\set{0,1}^{3s})^{\floor*{k/s}}$ and $\calC(y')= (\calC(y'_1),\dots,\calC(y'_{\lfloor k/s\rfloor}))\in ((\set{0,1}^{c_\delta})^s)^{\floor*{k/s}}$. We construct $\tc^{(s)}_{\Lag \ell}(x)$ by letting for all $i\in \set{1,\ldots, s}$ and for all $1\leq j \leq \floor*{k/s}$ such that $js+i-1\leq k$:
\[
\tc^{(s)}_{\Lag \ell}(x)_{js+i-1}=(\calC(y')_j)_i. 
\]

In other words, we read groups of $s$ symbols from $x\in \set{0,1}^k$, say $x|_{[js+1, (j+1)s]}$ and interpret each block of $s$ symbols as a step in $\tc^{(s)}$. Thus, we obtain a symbol in $\set{0,1}^{3s}$ which we then encode to a symbol in $(\set{0,1}^{c_\delta})^s$ using $\calC$. Then we "write" this $s$-character long string on the appropriate $s$-length segment of $\tc^{(s)}_{\Lag \ell}(x)$, i.e., from $(j+1)s$ to $(j+2)s-1$.

\begin{proposition}\label{prop:truncated-laggedTreeCode}
For all $x\in\set{0,1}^k$ the encoding procedure of $\tc^{(s)}_{\Lag \ell}(x)$ is online and $\tc^{(s)}_{\Lag \ell}(x)_{[1,s-1]}=\varepsilon$ where $\varepsilon$ is the empty string.
\end{proposition}
\begin{proof}
The online nature of the encoding is clear by construction. For the latter part notice that we don't begin "writing" the encoding of $\tc^{(s)}_{\Lag \ell}(x)$ until we see the first block of $s$ symbols and hence the encoding starts only from $\tc^{(s)}_{\Lag \ell}(x)_{s}$.
\end{proof}

\begin{lemma}\label{lemma:truncated-laggedTreeCode}
$\tc^{(s)}_{\Lag \ell}\colon \{0,1\}^{(\leq s^2)}
\to \{\{0,1\}^{c_\delta}\}^{(\leq s^2)}$ has $\ell$-lagged distance at least $\delta(\nicefrac{1}{2}-\nicefrac{3}{2a})$.
\end{lemma}

\begin{proof}
We need to show that for all $u,v\in \set{0,1}^k$ such that $b=k-\splt(u,v)\geq \ell$ we have \[\frac{\Delta(\tc^{(s)}_{\Lag \ell}(u),\tc^{(s)}_{\Lag \ell}(v))}{b}\geq \delta(\nicefrac{1}{2}-\nicefrac{3}{2a}).\] Notice that $|u'|-\splt(u',v')\geq \floor*{b/s}$ where $u'$ and $v'$ are analogues of $x'$ for $u$ and $v$ respectively. By distance of tree code $\tc^{(s)}$, we have $\Delta(\tc^{(s)}(u'), \tc^{(s)}(v'))\geq \frac{\floor*{b/s}}{2}$. By composition with $\calC$ we obtain that $\Delta(\tc^{(s)}_{\Lag \ell}(u),\tc^{(s)}_{\Lag \ell}(v))\geq \delta s (\floor*{b/s}\cdot \nicefrac{1}{2}-1)$. The $\floor*{b/s}-1$ term arises as we may not have enough length to encode $\calC(u')|_{\floor*{k/s}}$, i.e., it may be that $k<\floor*{k/s}s+s-1$.

Thus, $\frac{\Delta(\tc^{(s)}_{\Lag \ell}(u),\tc^{(s)}_{\Lag \ell}(v))}{b}\geq \frac{\delta s (\floor*{b/s}\cdot \nicefrac{1}{2}-1)}{b}\geq \delta(\nicefrac{1}{2}-\nicefrac{3}{2a})$ since $b \geq \ell \geq as$. 
\end{proof}

\paragraph{Step 2: untruncated code with constant $(\ell,O(\ell^2)$-lagged distance:}
We use interleaving copies of the above truncated-$\ell$-lagged distance $\delta(\nicefrac{1}{2}-\nicefrac{3}{2a})$ tree code to
construct a un-truncated-$(\ell,O(\ell^2))$-lagged distance $\delta(\nicefrac{1}{2}-\nicefrac{3}{2a})$ tree code.
Formally, let $\calC$ be the code from \cref{prop:ecc} with output alphabet
$\set{0,1}^{c_\delta}$ and distance $\delta$. For all $\ell \geq
as$ (where $s\geq s_\delta$ as in \cref{prop:ecc}), we describe an explicit-$(\ell,O(\ell^2))$-lagged tree code over the binary input alphabet and output alphabet $\{0,1\}^{2 c_\delta}$, \[\tc_{\Lag \ell}\colon \{0,1\}^{*}
\to {(\{0,1\}^{2 c_\delta})}^{*}.\]
Consider any $x\in \set{0,1}^k$. We describe the encoding of $x$, i.e., $\tc_{\Lag \ell}(x)$ below. Let $i\leq k$ be $j(s^2/2)+r$ where $r<s^2/2$.\footnote{We assume that $s$ is even for simplicity.} Then,
\[
\tc_{\Lag \ell}(x)_i = (\tc^{(s)}_{\Lag \ell}(x|_{[(j-1)s^2/2,i]}),\tc^{(s)}_{\Lag \ell}(x|_{[(j)s^2/2,i]})).
\]
Notice that if $j=0$ then $x|_{[(j-1)s^2/2,i]}$ is the empty string.

In other words, we place a $\ell$-lagged tree code $\tc^{(s)}_{\Lag \ell}\colon \{0,1\}^{s^2}
\to \{\{0,1\}^{c_\delta}\}^{s^2}$ at intervals of length $s^2/2$. So, for every $x\in \set{0,1}^k$ and $i>s^2/2$ there are two $\ell$-lagged tree codes whose encodings are "written" on $\tc_{\Lag \ell}(x)_i$.

\begin{proposition}\label{prop:untruncated-laggedTreeCode}
For all $x\in\set{0,1}^k$ the encoding procedure of $\tc_{\Lag \ell}(x)$ is online and $\tc_{\Lag \ell}(x)|_{[1,s-1]}=\varepsilon$ where $\varepsilon$ is the empty string. 
\end{proposition}
\begin{proof}
The online nature of the encoding is by construction.

$\tc_{\Lag \ell}(x)|_{[1,s-1]}=\varepsilon$  follows from $\tc^{(s)}_{\Lag \ell}(x)_{[1,s-1]}=\varepsilon$ as shown in \cref{prop:truncated-laggedTreeCode}. 
\end{proof}

\begin{lemma}\label{lemma:untruncated-laggedTreeCode}
$\tc_{\Lag \ell}\colon \{0,1\}^{*}
\to {(\{0,1\}^{2 c_\delta})}^{*}$ has $(\ell,O(\ell^2))$-lagged distance at least $\delta(\nicefrac{1}{2}-\nicefrac{3}{2a})$.
\end{lemma}

\begin{proof}
 We show that for all $u,v\in \set{0,1}^k$ such that $b=k-\splt(u,v)\geq\ell$ and $b\leq s^2/2$ we have $\frac{\delta(\tc_{\Lag \ell}(u),\tc_{\Lag \ell}(v))}{b}\geq \delta(\nicefrac{1}{2}-\nicefrac{3}{2a})$.
By construction we know that there is some $j$ such that $\splt(u,v)-js^2/2\leq s^2/2$ and $k-js^2/2\leq s^2$. Thus the encoding of $u$ and $v$ which include $\tc^{(s)}_{\Lag \ell}(u|_{[js^2/2,k]})$ and $\tc^{(s)}_{\Lag \ell}(v|_{[js^2/2,k]})$ differ on at least $\delta(\nicefrac{1}{2}-\nicefrac{3}{2a})b$ locations by virtue of $\tc^{(s)}_{\Lag \ell}$. Thus, $\tc_{\Lag \ell}$ is a $(\ell,s^2/2)$-lagged tree code with distance $\delta(\nicefrac{1}{2}-\nicefrac{3}{2a})$.
\end{proof}

\paragraph{Step 3: tree code with polylogarithmic output alphabet}
Finally, we superimpose several copies of the above $(\ell,O(\ell^2))$-lagged
tree code for varying $\ell$ to get an explicit family of constant distance $n$-truncated tree codes, $\set{\tc^{(n)}}_n$ with binary input alphabet and output alphabet of size $O(\polylog(n))$.
Formally, there exists a family of alphabets $\{\Gamma_n\}_{n}$ such that
$|\Gamma_n|= \polylog(n)$ with the following property. 
There exists a constant-distance tree code $\tc^{(n)}\colon\{0,1\}^{(\leq n)} \to
\Gamma_n^{(\leq n)}$ such that the encoding functions are constructible in time polynomial in $n$. $\tc^{(n)}$ is a superimposition of $\polylog(n)$ many copies of the $(\ell,O(\ell^2))$-lagged
tree code for varying $\ell$ such that for any two strings $u,v$, it is always the case that $\splt(u,v)$ lies in the stipulated range for one of these codes which witnesses the distance. We describe $\tc^{(n)}$ formally below.

For concreteness let us take $\delta=1/4$, $a=6$ and $\ell=6s$ in \cref{lemma:untruncated-laggedTreeCode}. Then for all $s\geq s_{\delta}$ we have a  $(6s, s^2/2)$-lagged tree code with distance $1/16$. In other words, for $\ell\geq 6s_{1/4}$ we have a distance-$1/16$-$(\ell,\ell^2/72)$-lagged tree code. 
To construct $\tc^{(n)}$ we will superimpose the encodings of $\tc_{\Lag \ell_i}$ for all $i\in \set{1, \ldots, j}$ where $\ell_1=6s_{1/4}$ and $\ell_{i+1}=\ell_i^2/72$
and $\ell_j^2/72=n$, thus $j=O(\log\log(n))$. \footnote{We ignore divisibility issues for the sake of simplicity.}
More precisely, for all $x\in \set{0,1}^k$ let 
$$\tc^{(n)}(x)_i=((x|_{[i-6s_{1/4},i]}),(\tc_{\Lag \ell_1}(x))_i,\ldots,(\tc_{\Lag \ell_j}(x))_i).$$ 
Hence, $\tc^{(n)}(x)\in \Gamma_n^k$ where $|\Gamma_n|=O(\polylog(n))$ and the encoding functions take time polynomial in $n$ since each coordinate is computable in time polynomial in $n$.

\begin{remark}\label{rem:honesttogodtreecode}
  The above construction yields a $n$-truncated tree code. To obtain a honest-to-god untruncated tree code, we observe that the Cohen-Haeupler-Schulman construction can be done incrementally in the following sense. \cref{prop:truncated-laggedTreeCode,prop:untruncated-laggedTreeCode,prop:binarytreecode} state that the initial part of the tree code encodings are empty. Hence, we might begin with a constant-sized alphabet (without printing the $\varepsilon$) and grow the alphabet as and when needed
  \end{remark}

\begin{proposition}\label{prop:binarytreecode}
For all $x\in \set{0,1}^n$, $i\in[n]$ and $g$ such that $\ell_g>6(i+1)$ we have $\tc_{\Lag\ell_g}(x)_i=\varepsilon$. Hence,  $\tc^{(n)}(x)_i\in \Gamma_i $ where $|\Gamma_i|=O(\polylog(i))$.
\end{proposition}
\begin{proof}
From \cref{prop:untruncated-laggedTreeCode} (with $a=6$) we know that $\tc_{\Lag\ell_j}(x)_i=\varepsilon$. In particular, we have that 
$$\tc^{(n)}(x)_i=((x|_{[i-6s_{1/4},i]}),(\tc_{\Lag \ell_1}(x))_i,\ldots,(\tc_{\Lag \ell_g}(x))_i,\varepsilon,\ldots, \varepsilon).$$
\end{proof}

Putting the above together we get that there exists a family of alphabets $\{\Gamma_n\}_{n}$ such that $|\Gamma_n|= O_\eta(\polylog(n))$ with the following property. Let $\Gamma=\cup_n \Gamma_n$. There exists a tree code $\tc\colon \{0,1\}^* \to
\Gamma^*$ such that the encoding functions $\phi_i$ satisfy
satisfy $\phi_i\colon \{0,1\}^i \to \Gamma_i$. Furthermore, the encoding functions $\phi_n$ are constructible in time polynomial in $n$.

\begin{theorem}\label{thm:binarytreecode}
There exists a family of alphabets $\{\Gamma_n\}_{n}$ such that
$|\Gamma_n|= O_\eta(\polylog(n))$ with the following property. Let $\Gamma=
\cup_n \Gamma_n$. There exists a tree code $\tc\colon \{0,1\}^* \to
\Gamma^*$ with distance $\nicefrac{1}{16}$ such that the encoding functions $\phi_i$ satisfy $\phi_i\colon \{0,1\}^i \to \Gamma_i$. Furthermore, the encoding functions $\phi_n$ are constructible in time polynomial in $n$.
\end{theorem}

\begin{proof}


The tree code we consider is the tree code $\tc$ as mentioned above the theorem statement.
We have already shown that $\tc\colon \{0,1\}^* \to \Gamma^*$ is such that the encoding functions $\phi_i$ satisfy $\phi_i\colon \{0,1\}^i \to \Gamma_i$. Furthermore, the encoding functions $\phi_n$ are constructible in time polynomial in $n$.

Next, we show that the distance of $\tc^{(n)}$ is at least $1/16$. Consider $u,v\in \set{0,1}^k$ and let $b=k-\splt(u,v)$. If $b\leq 6s_{1/4}$ then for all $i\in \set{\splt(u,v)+1,k}$ we have $(\tc^{(n)}(u))_i\neq (\tc^{(n)}(v))_i$ as they differ in the first coordinate. Now, suppose $\ell_i\leq b\leq \ell_i^2/72$ for some $1\leq i\leq j$. Then, by virtue of $\tc_{\Lag \ell_i}$ we know that $\frac{\Delta(\tc^{(n)}(u),\tc^{(n)}(v))}{b}\geq 1/16$. In any case, we have $\delta_{\tc^{(n)}}\geq 1/16$.
\end{proof}

Notice that if we use the boosting trick mentioned in \cref{thm:MDScodes_r/(r+s)} we obtain a distance $r/(r+s)$ tree code over the integers for any $r,s\in\naturals$. If we use this tree code in the procedure detailed in the proof of \cref{thm:binarytreecode} we obtain $\tc^{(n)}$ with distance $\left(1-O\left(\frac{r}{r+s}-O\left(\frac1a\right)\right)\delta\right)$ where $\delta$ is the distance of $\calC$ in \cref{prop:ecc} and $a=\ell/s$ is as describe in the proof of \cref{lemma:truncated-laggedTreeCode}. Hence, by choosing $r/(r+s)$ and $\delta$ close to $1$ and $a$ large enough we have established \cref{thm:binarytreecodearbitrarydistance}, which is restated below for convenience.

\binarytreecodearbitrarydistance*

\section{Open Questions} 
Recall that the $\tcz$ of \cref{thm:inttreecode} used in \cref{cor:largealph} was an explicit tree code over the integers
$\tcz\colon \integers^{\integers} \to (\integers^2)^{\integers}$
with distance at least $1/2$ satisfying the property that for any $n
\in \pintegers$ and $i \in \{0,1\dots, n\}$, we have
$|(\tcz(a_0,\dots, a_n))_i| \leq 2^n \max_{j} |a_j|^2$. This in turn was used in step $1$ of the reduction process to small alphabets. 
If instead our $\tcz$ had the property that $|(\tcz(a_0,\dots, a_n))_i)| \leq 2^{\polylog(n)} \max_{j} |a_j|^2$, then the above reduction process to small alphabets would actually yield constant-sized alphabets. 
\footnote{Strictly speaking, $|(\tcz(a_0,\dots, a_n))_i)| \leq 2^{f(n)} \max_{j} |a_j|^2$ yields a tree code whose (relative) $(2f(n), nf(n))$-lagged distance is constant. Thus, any function $f(n)$ would suffice for constant-sized alphabets as long as we can show that there exists a constant $c$ such that $\underbrace{g\circ g\circ \ldots \circ g}_{c-\text{times}}(n)=O(\log\log(n))$ where $g(nf(n))=2f(n)$. }
Since, $\tcz$ of \cref{thm:inttreecode} was obtained using a totally-non-singular lower-triangular matrix (specifically the Pascal Matrix), we are led to the following question about the existence of such $n\times n$ matrices with entries only $2^{\polylog(n)}$ which if true would imply tree codes with constant distance and  constant-sized alphabets. Furthermore, if the matrices are explicit then we would get explicit tree codes with constant distance and constant-sized alphabets.

 \begin{question}[Existence of totally-non-singular lower-triangular matrices over small alphabets]
 Are there $n\times n$ totally-non-singular lower-triangular matrices as defined in \cref{defn:TraingularTotallyNonSingular} with entries of size $O(2^{\polylog(n)})$. Further, are there such matrices which are explicit?
 (Pudl\'{a}k~\cite[Problem~1]{Pudlak2016} states this question asking whether the entries of such a matrix can be $\mathrm{poly}(n)$.)
 \end{question}

The next question asks about the decoding process to accompany the constant-distance tree codes of \cref{thm:binarytreecodearbitrarydistance}.

\begin{question}[Efficient decoding of tree codes]
Is there an efficient decoding procedure for the tree codes of \cref{thm:binarytreecodearbitrarydistance} which corrects a constant fraction of errors?
\end{question}

There has been recent work on this question by Narayanan and Weidner~\cite{NarayananW2020}. The authors use a randomized polynomial time decoding algorithm to correct errors which scale roughly as the block length to the three-fourths power, falling short of the constant fraction error correction guaranteed by the constant distance.




\section*{Acknowledgements} We thank Madhu Sudan for various discussions while preparing this exposition, in particular, we thank him for suggesting the alternate construction of distance $\nicefrac12$-tree codes over the integers mentioned after \cref{clm:cgeqr}.

{\small
\bibliographystyle{prahladhurl}
\bibliography{treecodes-bib}
}

\end{document}